\documentclass[12pt,reqno,fleqn]{amsart}
\usepackage{latexsym}
\usepackage{amssymb}
\usepackage{amsxtra}
\usepackage{amsmath}
\usepackage{amsfonts}
\newtheorem{theorem}{Theorem}[section]
\newtheorem{lemma}[theorem]{Lemma}

\usepackage{amssymb}
\usepackage{mathrsfs}
\usepackage{amsmath}
\usepackage{enumerate}

\makeatother       % '@' is restored as a "non-letter" character for TeX

\makeatletter      % '@' is now a normail "letter" for TeX
\@addtoreset{equation}{section}
\makeatother       % '@' is restored as a "non-letter" character for TeX

\begin{document}

\title[constrained discrete KP]{Virasoro type algebraic structure hidden in the constrained discrete KP hierarchy}
\author{Maohua Li$^{1,2},$  Chuanzhong Li$^2$, Keilei Tian$^{1,3}$, Jingsong
He$^2$$^*$,Yi Cheng$^1$ }\dedicatory {
\mbox{}\hspace{-2.5cm}1. Department of Mathematics, USTC, Hefei, 230026 Anhui, P.R.China\\
2. Department of Mathematics, Ningbo University, Ningbo, 315211 Zhejiang, P.R.China\\
\mbox{}\hspace{-.5cm} 3. Department of Physics and Astronomy,
University of Turku, 20014 Turku, Finland \\}
\date{}

\thanks{$^*$ Corresponding author: hejingsong@nbu.edu.cn}

%%%%%%%%%%%%%%%%%%%%%%%%%%%%%%%%%%%%%%%%%%%%%%%%
\begin{abstract}
In this paper, we construct the additional symmetries of one-component
constrained discrete KP (cdKP) hierarchy, and then prove that the
algebraic structure of  the symmetry flows is the positive half
of Virasoro algebra.
\end{abstract}

%%%%%%%%%%%%%%%%%%%%%%%%%%%%%%%%%%%%%%%%%%%%%%%%

\maketitle
%\tableofcontents \allowdisplaybreaks \setcounter{section}{0}

%\begin{center}
\section{ Introduction}
%\end{center}

In the past few years, lots of attention have been given to the study of
 Kadomtsev-Petviashvili (KP) hierarchy \cite{dkjm,dl1} in the field of integrable systems.
The Lax pairs, Hamiltonian structures, symmetries and conservation
laws, the N-soliton, tau function,  the gauge transformation,
reductions etc. of the KP hierarchy and its sub-hierarchies have
been discussed. A specific interesting aspect of the research of the
KP hierarchy is additional symmetry\cite{dl1,os1,asv2,tu07lmp,hetian07,jipenghe2010}.
Additional symmetries are special symmetries which are not contained in the KP
hierarchy and do not commute  with each other. The additional
symmetry flows of the KP hierarchy form an infinite dimensional
algebra $W_{1+\infty}$\cite{asv2}. Recently, there are
several new results   about partition function in the matrix models
and Seiberg-Witten theory associated with additional symmetries,
string equation and Virasoro constraints of the KP hierarchy
\cite{Morozov1994,Morozov1998,Alexandrov, Mironov2008,Aratyn2003}.

There are several sub-hierarchies of the KP by considering different
reduction conditions on the Lax operator $L$. One of them is called
constrained KP (cKP) hierarchy \cite{kss,chengyiliyishenpla1991,
chengyiliyishenjmp1995} by setting the  Lax operator as  $L=\partial
+\sum _{i=1}^m \Phi_i \partial ^{-i}\Psi _i$. Here $ \Phi_i $ is
an eigenfunction and $\Psi _i $ is an adjoint eigenfunction of the
cKP hierarchy. The cKP hierarchy contains a large number of interesting
 soliton equations. The basic idea of this procedure is
so-called symmetry constraint\cite{kss,chengyiliyishenpla1991,
chengyiliyishenjmp1995}. The negative part of the Lax operator of the
constrained KP, i.e. $\sum _{i=1}^m \Phi_i \partial ^{-i}\Psi _i$,
is a generator\cite{dl1} of the additional symmetry of the KP hierarchy.
This observation  inspires  the study\cite{aratyn1997} of the
additional symmetries of the cKP hierarchy. However, the
additional symmetry flows of the KP hierarchy are not consistent with the
special form of the  Lax operator for the cKP hierarchy autonomously,
and so it is highly non-trivial  to find  a suitable form of additional
symmetry  flows for this sub-hierarchy.  The correct additional symmetry flows are given  by means of
a crucial modification  associated with a complicated  operator $X_k$\cite{aratyn1997}.
Very recently, by a further modification of the additional flows,
 the additional symmetries of
 the constrained BKP and constrained CKP hierarchies are given in references
 \cite{tianhe2011,tu_shenJMP2011}. In addition to the above-mentioned
 integrable systems, discrete system such as Toda hierarchy also has interesting algebraic structure
of additional symmetry\cite{asv2,jipenghe2011}.

 At the same time the discovery of Hirota's bilinear difference equation \cite{Hirota1981}
 attracted much interest in looking for other integrable discrete equations and systems.
This famous 3-dimensional difference equation is known to provide a canonical integrable
discretization for most important types of soliton equations.
There are several difference kinds of the discrete hierarchies including
differential-difference
 KP hierarchy, semi-discrete systems, full discrete equations and so on. The differential-difference
 KP (dKP) hierarchy \cite{Kupershimidt,Iliev,czli2012} defined by the difference operator
 $\triangle$ is one interesting object  of the integrable discrete systems. Note that,
 the additional symmetry of dKP hierarchy and it's Sato B\"acklund
 transformations  have been given in reference\cite{LiuS2}.
So it is worthy to find the hidden algebraic structure in the
constrained discrete KP(cdKP) hierarchy using the additional symmetry flows,
which is the main purpose of us. This will be done by following four
steps:
\begin{itemize}
\item
1) show the inconsistency between the additional symmetry
flows eq.(\ref{additionalofdkp}) of the dKP hierarchy and the
 Lax form eq.(\ref{1cdKPlaxeq}) of the cdKP hierarchy;
\item
2) modify the additional flows to sort out  this inconsistency;
\item
3) prove that the modified additional flows are symmetry of the cdKP
hierarchy;
\item
4) identify the algebraic structure of the additional symmetry
flows.
\end{itemize}
The crucial step of this process is to find a suitable modification
of the additional flows of the dKP hierarchy such that these flows
are consistent with the form of the Lax pair of the cdKP hierarchy.

The paper is organized as follows. Some basic results of of dKP
hierarchy and the additional symmetry of dKP hierarchy are
summarized in Section \ref{section2}. After introducing of a definition of the
Lax equation of cdKP hierarchy, the additional symmetry flows of the
cdKP hierarchy
 are defined properly by means of a crucial modification
 from the additional symmetry flows of the dKP hierarchy in Section \ref{section3}.
Next the Virasoro type algebraic structure of these additional symmetry flows
is also identified by a straightforward calculation in Section \ref{section4}.
Section \ref{conclusion} is devoted to conclusions and discussions.

%%%%%%%%%%%%%%%%%%%%%%%%%%%%%%%%%%%%%%%%%%%%%%%%%%%%%%%%%%%%%%%%%%%%%%%%%%%%%%%%%%%%%%%%%%%%%%%%%%%%%%%%%%%%%%%%%
%%%%%%%%%%%%%%%%%%%%%%%%%%%%%%%%%%%%%%%%%%%%%%%%%%%%%%%%%%%%%%%%%%%%%%%%%%%%%%%%%%%%%%%%%%%%%%%%%%%%%%%%%%%%%%%%%
%\begin{center}
\section{The dKP hierarchy and it's additional symmetry}\label{section2}
%\end{center}
\setcounter{section}{2} Let us briefly recall some basic facts
about the dKP hierarchy according to reference \cite{Iliev}.
Firstly
a space $F$, namely
\begin{equation}
F=\left\{ f(n)=f(n,t_1,t_2,\cdots,t_j,\cdots);  n\in\mathbb{Z}, t_i\in\mathbb{R}
\right\}
\end{equation}
is defined for the space of the discrete KP hierarchy.
And $\Lambda,\triangle$ are denote for the shift operator and the
difference operator, respectively. Their actions on function $f(n)$ are
define for
\begin{equation}\Lambda f(n)=f(n+1)
\end{equation}
and
\begin{equation}\triangle f(n)=f(n+1)-f(n)=(\Lambda -I)f(n)
\end{equation} respectively,
where $I$ is the identity operator.

For any $j\in\mathbb{Z},$ the Leibniz  rule of $\triangle$ operation is,

\begin{equation}\triangle^j\circ
f=\sum^{\infty}_{i=0}\binom{j}{i}(\triangle^i
f)(n+j-i)\triangle^{j-i},\hspace{.3cm}
\binom{j}{i}=\frac{j(j-1)\cdots(j-i+1)}{i!}.\label{81}
\end{equation}
So an associative ring $F(\triangle)$ of formal pseudo
difference operators is obtained, with the operation $``+"$ and $``\circ"$, namely
$F(\triangle)=\left\{R=\sum_{j=-\infty}^d f_j(n)\triangle^j,
f_j(n)\in R, n\in\mathbb{Z}\right\},
$
and denote $R_+:=\sum_{j=0}^d f_j(n)\circ\triangle^j$ as the
positive projection of $R$ and by $R_-:=\sum_{j=-\infty}^{-1}
f_j(n)\circ \triangle^j$, the negative projection of $R$. The
adjoint operator to the $\triangle$ operator is given by
$\triangle^*$,
\begin{equation}
\triangle^* \circ f(n)=(\Lambda^{-1}-I)f(n)=f(n-1)-f(n),
\end{equation}
where $\Lambda^{-1} f(n)=f(n-1)$, and the corresponding ``$\circ$"
operation is
\begin{equation}
\triangle^{*j}\circ
f=\sum^{\infty}_{i=0}\binom{j}{i}(\triangle^{*i}f)(n+i-j)\triangle^{*j-i}.
\end{equation}
Then the adjoint ring $F(\triangle^*)$ to the
$F(\triangle)$ is obtained, and the formal adjoint to $R\in F(\triangle)$ is defined
by $R^*\in F(\triangle^*)$ as $R^*=\sum_{j=-\infty}^d
\triangle^{*j}\circ f_j(n)$. The $"*"$ operation satisfies rules as
$(F\circ G)^*=G^*\circ F^*$ for two operators and $f(n)^*=f(n)$ for
a function.

         The discrete KP (dKP) hierarchy \cite{Iliev} is a family of evolution equations depending on
infinitely many variables $t=(t_1,t_2,\cdots)$
\begin{equation} \label{floweq}
\frac{\partial L}{\partial t_k}=[B_k, L],\ \ \ B_k:=(L^k)_+,
\end{equation}
where $L$ is a general first-order pseudo difference operator(PDO)
\begin{equation} \label{laxoperatordkp}
L(n)=\triangle + \sum_{j=1}^{\infty} f_j(n)\triangle^{-j}.
\end{equation}
 $B_m=(L^m)_+=\sum^m_{j=0}a_j(n)\triangle^j$, i. e.  $(L^m)_+$ is the
non-negative projection of $L^m$, and $(L^m)_-=L^m-(L^m)_+$ is the
negative projection of $L^m$. The Lax operator in
eq.(\ref{laxoperatordkp}) can be generated by a dressing operator
\begin{equation}
W(n;t)=1+\sum^\infty_{j=1}w_j(n;t)\triangle^{-j}.
\end{equation}
through
\begin{equation}
L=W \circ \Delta \circ W^{-1}.
\end{equation}
Further the flow equation
(\ref{floweq}) is equivalent to the so-called Sato equation,
\begin{equation}\label{tkaction}
\partial_{t_k}W=-(L^{k})_-\circ W.
\end{equation}

Now we introduce the additional symmetry flows\cite{LiuS2} of the dKP hierarchy as following.
Set
\begin{equation}
\Gamma_\Delta=\sum_{i=1}^{\infty}(it_i\Delta^{i-1}+{(-1)}^{i-1}n\Delta^{i-1}),
\end{equation} and it  is easy to find the following formula
\begin{equation}\label{commute}
[\partial_{t_k}-\Delta^k,\Gamma_{\Delta}]=0.
\end{equation}
Define another operator
\begin{equation}\label{}
M_{\Delta}=W \circ \Gamma_{\Delta}\circ W^{-1}.
\end{equation}
There are the following commutation relations
\begin{equation}
[\Delta,\Gamma_{\Delta}]=1,[L,M_{\Delta}]=1,
\end{equation}
which can be verified by a straightforward calculation.
By using the Sato equation, the isospectral flow of the $M_{\Delta}$ operator
is given by
 \begin{equation}
\partial_{t_k}M_{\Delta}=[L^k_+,M_{\Delta}].
\end{equation}  More generally,
\begin{equation}
\partial_{t_k}(M_{\Delta}^mL^l)=[L^k_+,M_{\Delta}^mL^l].
\end{equation}

Based on the above preparation, the additional symmetry flows\cite{LiuS2}
of the dKP hierarchy
are define by their actions on the dressing operator
 \begin{equation}
\overline{\partial_{l,m}}W=-(M_{\Delta}^mL^l)_-\circ W,
\end{equation}
or equivalently  on the Lax operator
\begin{eqnarray}
\overline{\partial_{l,m}}L&=&[-(M_{\Delta}^mL^l)_-,L]
\label{additionalofdkp},
\end{eqnarray} where $\overline{\partial_{l,m}}$ denotes
the derivative with respect to an additional new
variable $t_{lm}^*$. The more general actions of the additional symmetry flows
of the dKP are given by
\begin{equation}\label{ML}
 \overline{\partial_{l,m}}M_{\Delta}^nL^k=
 [-(M_{\Delta}^mL^l)_-,M_{\Delta}^nL^k].
\end{equation}

As the end of this section, we would like to point out two important
 technical identities as followings. For two pseudo-difference
 operators $X_i=f_i\Delta^{-1}g_i,i=1,2,$ we have
\begin{equation}
X_1\circ X_2=X_1(f_2)\Delta^{-1}g_2+f_1\Delta^{-1}X_2^{*}(g_1).
\label{operatorX1X2}
\end{equation}
For a pure-difference operator $K$ and two arbitrary smooth functions ($q, r$),
we have
\begin{equation}\label{operatork}
[K,q\Delta^{-1}r]_-=K(q)\Delta^{-1}r-q\Delta^{-1}K^*(r).
\end{equation}
The usual versions(non-discrete) of them are given by eq.(A.3) and
eq.(A.4) in reference\cite{aratyn1997}.

\section{Additional symmetry flows of the cdKP hierarchy}\label{section3}

 The one-component cdKP hierarchy is defined by following Lax equation
\begin{equation}\label{1cdKPlaxeq}
\frac{\partial L}{\partial t_l} = [(L^l)_+,L], l=1,2,\cdots,
\end{equation}
associated with a special Lax operator
\begin{equation} \label{laxofcdkp2}
L = (L)_{+} + q(t)\triangle^{-1}r(t),
\end{equation}
 and  $q(t)$ is an eigenfunction, $r(t)$ is an adjoint
 eigenfunction. The eigenfunction and adjoint eigenfunction {$q(t),r(t)$} are
important dynamical variables in the cdKP hierarchy.
Using identity eq.(\ref{operatork}), one can check
 Lax equation (\ref{1cdKPlaxeq}) is consistent with the evolution
 equations of the eigenfunction(or adjoint eigenfunction)
\begin{equation}
q_{t_m} = B_mq,\quad  r_{t_m} = -B_m^*r,\quad  B_m=(L^m)_{+}, \forall m \in N.
\end{equation}
Therefore  the cdKP hierarchy in eq.(\ref{1cdKPlaxeq}) is well defined.
And Eq.(\ref{1cdKPlaxeq}) implies that the Sato equation of cdKP hierarchy is
\begin{equation} \label{satoofcdkp}
\partial_l W = -(L^l)_-\circ W,
\end{equation}
where $\partial_l=\frac{\partial }{\partial t_l}$.

The central task of this section is to find the additional symmetry
flows of the cdKP hierarchy, which can be realized by three steps as
we have mentioned in the introduction. As usual calculation of the
infinitesimal analysis, the desired action of the additional
symmetry flows on the Lax operator $L$ of the cdKP hierarchy should be
\begin{equation}\label{araty}
(\widetilde{\partial_{\tau}}L)_-=\widetilde{\partial_{\tau}}(q(n,t))\triangle^{-1}r(n,t)
+q(n,t)\triangle^{-1}\widetilde{\partial_{\tau}}(r(n,t)).
\end{equation}
However, according to the definition eq.(\ref{additionalofdkp}) of the dKP hierarchy,
the action of original
additional flows of the cdKP hierarchy  is  expressed by
\begin{equation}\label{additional}
(\overline{\partial_{k,1}} L)_-=[(M_{\Delta}L^k)_+,L]_-+(L^k)_-.
\end{equation}
The following lemma and  eq.(\ref{operatork}) show that it can not be rewritten as
the desired form eq.(\ref{araty}) except $k=0,1,2$. Specifically, $k=3$, from $L^3_{-}$
\begin{equation}\label{L3}
(L^3)_-=L^2(q)\Delta^{-1}r+L(q)\Delta^{-1}L^*(r)+q\Delta^{-1}L_2^*(r),
\end{equation} we can find the middle term can not be rewritten as the form of eq.(\ref{araty}).
This demonstrates obviously the
inconsistency  between the additional symmetry
flows eq.(\ref{additionalofdkp}) of the dKP hierarchy and the
 Lax form eq.(\ref{1cdKPlaxeq}) of the cdKP hierarchy.
\begin{lemma}The Lax operator $L$ of constrained dKP hierarchy given by (\ref{laxofcdkp2}) satisfied the relation of
\begin{equation}\label{Lk}
(L^k)_-=\sum_{j=0}^{k-1}L^{k-j-1}(q)\Delta^{-1}{(L^*)}^j(r).
\end{equation}
where $L(q)=(L)_{+}(q) + q(t)\triangle^{-1}(r(t)q).$
\end{lemma}
\begin{proof}It can be reduced by induction with the help of
technical identity in eq.(\ref{operatorX1X2}).  We omit it here.
\end{proof}

To overcome the inconsistency, we shall introduce an operator
$Y_k$  to modify the additional symmetry
eq.(\ref{additionalofdkp}) of the dKP hierarchy. The following
lemmas implied by identity (\ref{operatorX1X2},\ref{operatork}) are
necessary.
\begin{lemma}
\begin{equation}\label{operatorXL}
[X,L]_-=\sum_{k=1}^{l}[M_k\Delta^{-1}L^*(N_k)-L(M_k)\Delta^{-1}N_k]+[X(q)\Delta^{-1}r-q\Delta^{-1}X^*(r)],
\end{equation} with definitions (\ref{laxofcdkp2}) and
\begin{equation}\label{x}
X=\sum_{k=1}^{l}M_k\Delta^{-1}N_k.
\end{equation}
\end{lemma}

We now introduce a pseudo-difference operator  $Y_k$,
\begin{eqnarray}
Y_k&=&\sum_{j=0}^{k-1}[j-\frac{1}{2}(k-1)]L^{k-1-j}(q)
\Delta^{-1}(L^*)^j(r),k\geq 2,\label{y} \\
Y_k&=&0,k=-1,0,1.\label{y12}
\end{eqnarray}
and have the following property.
\begin{lemma}
The action of flows $\partial_l$ of the cdKP hierarchy  on the
$Y_{k}$ is
\begin{equation}\label{ykderivat}
\partial_{l}
Y_{k}=[(L^l)_+,Y_{k}]_-.
\end{equation}
\end{lemma}
\begin{proof}
\begin{eqnarray*}
\partial_{l}Y_{k}&=&\partial_{l}\sum_{j=0}^{k-1}[j-\frac{1}{2}(k-1)]L^{k-1-j}(q)\Delta^{-1}(L^*)^j(r)\\\nonumber
&=&\sum_{j=0}^{k-1}[j-\frac{1}{2}(k-1)]\{\partial_{l}(L^{k-1-j}(q))\Delta^{-1}(L^*)^j(r)+L^{k-1-j}(q)\Delta^{-1}\partial_{l}((L^*)^j(r))\}\\\nonumber
&\stackrel{by(\ref{operatork})}{==}&[(L^l)_+\circ \sum_{j=0}^{k-1}[j-\frac{1}{2}(k-1)]L^{k-1-j}(q)\Delta^{-1}(L^*)^j(r)]_-\\
&&-[(\sum_{j=0}^{k-1}[j-\frac{1}{2}(k-1)]L^{k-1-j}(q)\Delta^{-1}(L^*)^j(r))\circ (L^l)_+]_-\\
&=&[(L^l)_+,(\sum_{j=0}^{k-1}[j-\frac{1}{2}(k-1)]L^{k-1-j}(q)\Delta^{-1}(L^*)^j(r))]_-\\
&=&[(L^l)_+,Y_{k}]_-.
\end{eqnarray*}
\end{proof}
The first nontrivial example of (\ref{y}) is given by
\begin{equation}\label{y2}
Y_2=[-\frac{1}{2}L(q)\Delta^{-1}r+\frac{1}{2}q\Delta^{-1}L^*(r)],
\end{equation}for $k=2$.
And
\begin{equation}\label{xy2}
[Y_2,L]_-=-(L^3)_-+\frac{3}{2}[L^2(q)\Delta^{-1}r+q\Delta^{-1}(L^*)^2(r)] \\
+[Y_2(q)\Delta^{-1}r-q\Delta^{-1}{Y_2}^*(r)].
\end{equation}

Further, the following expression of $[Y_{k-1},L]_-$ is also
necessary to define the additional flows of the cdKP hierarchy.
\begin{lemma}
The Lax operator $L$ of constrained dKP hierarchy and $Y_{k-1}$ has the following relation,
\begin{equation}\label{ykl}
[Y_{k-1},L]_-=-(L^k)_-+\frac{k}{2}[q\Delta^{-1}{(L^*)}^{k-1}(r)+L^{k-1}(q)\Delta^{-1}r]+Y_{k-1}(q)\Delta^{-1}r-q\Delta^{-1}Y_{k-1}^*(r)
\end{equation}
\end{lemma}
\begin{proof}
\begin{eqnarray*}
[Y_{k-1},L]_-&=&[\sum^{k-2}_{j=0}[j-\frac{1}{2}(k-2)]L^{k-2-j}(q)\Delta^{-1}{L^*}^j(r),L]_-\\\nonumber
&\stackrel{by (\ref{operatorXL})}{==}&\sum^{k-2}_{j=0}[j-\frac{1}{2}(k-2)]L^{k-2-j}(q)\Delta^{-1}{(L^*)}^{j+1}(r)
-\sum^{k-2}_{j=0}[j-\frac{1}{2}(k-2)]L^{k-1-j}(q)\Delta^{-1}{L^*}^j(r)\\\nonumber
&+&Y_{k-1}(q)\Delta^{-1}r-q\Delta^{-1}{Y_{k-1}}^*(r)\\\nonumber
&=&-\sum^{k-2}_{j=1}L^{k-1-j}(q)\Delta^{-1}{(L^*)}^j(r)+(\frac{k}{2}-1)[q\Delta^{-1}{(L^*)}^{k-1}(r)+L^{k-1}(q)\Delta^{-1}r]\\\nonumber
&+&Y_{k-1}(q)\Delta^{-1}r-q\Delta^{-1}Y_{k-1}^*(r)\\\nonumber
&=&-(L^k)_-+\frac{k}{2}[q\Delta^{-1}{(L^*)}^{k-1}(r)+L^{k-1}(q)\Delta^{-1}r]+Y_{k-1}(q)\Delta^{-1}r-q\Delta^{-1}Y_{k-1}^*(r)
\end{eqnarray*}
\end{proof}

Putting together (\ref{additional}) and (\ref{ykl}), we define the
additional flows of the cdKP hierarchy  as
 \begin{equation}\label{tkflow}
\partial_{k}^*L=[-(M_{\Delta}L^k)_-+Y_{k-1},L],
\end{equation}
where $\partial_{k}^*=\overline{\partial_{k,1}}$ and
$Y_{l-1}=0$, for $l=0,1,2$, such that the right-hand side of
(\ref{tkflow}) is in the form of (\ref{araty}). It must be mentioned
that the additional flows $\partial_{k}^*$ of cdKP hierarchy is nothing but the additional flows $\overline{\partial_{k,1}}$ of
the dKP hierarchy for $k=0,1,2$. Generally,
\begin{equation}\label{MLK}
\partial_{k}^*(M_{\Delta}L^l)=[-(M_{\Delta}L^k)_-+Y_{k-1},M_{\Delta}L^l].
\end{equation}

Now we calculate the action of the additional flows (\ref{tkflow})
on the eigenfunction $q$ and $r$ of the cdKP hierarchy.
\begin{theorem}\label{symmetre}
The acting of additional flows of constrained dKP hierarchy on the eigenfunction $q$ and $r$ are
\begin{equation}\label{BAfunction}
\begin{split}
{\partial_{k}^*q}&=(M_\Delta
L^k)_+(q)+Y_{k-1}(q)+\frac{k}{2}L^{k-1}(q),\\
 {\partial_{k}^*r}&=-(M_\Delta
L^k)^*_+(r)-Y_{k-1}^*(r)+\frac{k}{2}{(L^*)}^{k-1}(r).
\end{split}
\end{equation}
\end{theorem}
\begin{proof}
Substitution of (\ref{ykl}) to negative part of (\ref{tkflow}) shows
\begin{eqnarray}\label{tkflow2}
\begin{split}
{(\partial_{k}^*L)}_-&=(M_\Delta
L^k)_+(q)\Delta^{-1}(r)-q\Delta^{-1}(M_\Delta L^k)_+^*(r)\\
&+Y_{k-1}(q)\Delta^{-1}r-q\Delta^{-1}Y_{k-1}^*(r)
+\frac{k}{2}q\Delta^{-1}(L^*)^{k-1}(r)+\frac{k}{2}L^{k-1}(q)\Delta^{-1}r.
\end{split}
\end{eqnarray}

On the other side,
\begin{equation}\label{tkl}
{(\partial_{k}^*L)}_-=(\partial_{k}^*q)\Delta^{-1}r+q\Delta^{-1}{(\partial_{k}^*r)}.
\end{equation}
Comparing right hand sides of (\ref{tkflow2}) and (\ref{tkl})
implies  the action of additional flows on the eigenfunction and the
adjoint eigenfunction (\ref{BAfunction}).
\end{proof}

And the case  of (\ref{BAfunction}) with $k=3$ is
\begin{eqnarray}
\begin{split}
{\partial_{3}^*q}&=(M_\Delta
L^3)_+(q)+Y_{2}(q)+\frac{3}{2}L^{2}(q),\\
{\partial_{3}^*r}&=-(M_\Delta
L^3)^*_+(r)-Y_{2}^*(r)+\frac{3}{2}{L^*}^{2}(r).
\end{split}
\end{eqnarray}

Next we shall prove the commutation relation between the additional
flows $\partial^*_k$ of cdKP hierarchy and the original flows $\partial_{l}$ of
it.
\begin{theorem}
The additional flows of $\partial^*_k$ commute with all $\partial_{l}$ flows of the  cdKP hierarchy.
\end{theorem}
\begin{proof}
According the action of  $\partial^*_k$ and $\partial_{l}$ on the
Sato operator $W$ (\ref{satoofcdkp},\ref{tkflow}) and
(\ref{ykderivat}), then
\begin{eqnarray*}
[\partial^*_k,\partial_{l}]W &=& -\partial^*_k(L_-^l
W)-\partial_l[-(M_\Delta L^k)_-+Y_{k-1}]W \\
&=&(-\partial_k^*L_-^l)W-L^l_-\partial_k^*W-[(M_\Delta
L^k)_--Y_{k-1}]L_-^lW \\
 &&+[L^l_+,M_\Delta L^k]_-W-(\partial_l Y_{k-1})W\\
&=&[L^l_-,-Y_{k-1}]_-W+[-Y_{k-1},L^l]_-W-(\partial_l Y_{k-1})W \text{\quad by Jacobi identity}\\
&\stackrel{(\ref{ykderivat})}{==}&[L^l_+,-Y_{k-1}]_-W-(\partial_l Y_{k-1})W\\
&=&0.
\end{eqnarray*}
Therefore, the additional  flows $\partial^*_k$ commute with all flows $\partial_l$  of the cdKP hierarchy.
\end{proof}
\noindent\textbf{Remark:} This theorem  implies that the additional
flows $\partial^*_k$ (\ref{tkflow}) are  symmetry flows of the cdKP
hierarchy.

\section{Virasoro type algebraic structure of the additional symmetry
flows}\label{section4}

In this section, we shall discuss the algebraic structure of the
additional symmetry  flows of the cdKP hierarchy. For this end, we
need the actions of $\partial_{l}^*$ on $Y_k$ and $L$.

Taking into account $Y_{k-1}=0$  for $k=0,1,2 $, then
Eqs.(\ref{BAfunction}) becomes
\begin{eqnarray}\label{PL123}
\begin{split}
{\partial_{0}^*q}&=(M_\Delta
)_+(q),&
{\partial_{0}^*r}&=-(M_\Delta
)^*_+(r),\\
{\partial_{1}^*q}&=(M_\Delta
L)_+(q)+\alpha q,&
{\partial_{1}^*r}&=-(M_\Delta
L)^*_+(r)+\beta r, \alpha +\beta=1,\\
{\partial_{2}^*q}&=(M_\Delta
L^2)_+(q)+ L(q),&
{\partial_{2}^*r}&=-(M_\Delta
L^2)^*_+(r)+L^*(r).
\end{split}
\end{eqnarray}
We can rewrite the (\ref{PL123}) for
\begin{eqnarray}\label{PLqr}
\begin{split}
\partial_{l}^*q&=(M_\Delta
L)_+(q)+\frac{1}{2}l L^{l-1} q,\\
\partial_{l}^*r&=-(M_\Delta
L)^*_+(r)+\frac{1}{2}l {(L^*)}^{l-1} r,
\end{split}
\end{eqnarray} with $l=0,1,2.$

Because of $\partial_{l}^*=\overline{\partial_{l,1}},l=0,1,2$ as
  mentioned above, we have the following lemma.
\begin{lemma}\label {LqLr}
The additional  flows $\partial_{l}^*$ of cdKP hierarchy have the following relations for $l=0,1,2$ and $k \geq 0$, namely,
\begin{eqnarray}\label{lqstar}
\begin{split}
{\partial_{l}^* L^k(q)}&=(M_\Delta
L^l)_+(L^k(q))+(k+\frac{l}{2})L^{k+l-1}(q),\\
{\partial_{l}^* {(L^*)}^k(r)}&=-(M_\Delta
L^l)^*_+{(L^*)}^k(r)+(k+\frac{l}{2}){(L^*)}^{k+l-1}(r).
\end{split}
\end{eqnarray}
\end{lemma}
\begin{proof}
It is easy to get this  by using  (\ref{ML}), (\ref{PLqr}) and the
relation
${\overline{\partial_{l,1}}(L^k(q))}=(\overline{\partial_{l,1}}(L^k))(q)+L^k
\overline{\partial_{l,1}}(q)$.
\end{proof}

Moreover, the action of
 $\partial^*_l$ on $Y_{k}$ is given by the following lemma.
\begin{lemma}
The actions on $Y_{k}$ of the additional  symmetry flows
$\partial^*_l$ of the cdKP hierarchy are
\begin{eqnarray}\label {PL}
\partial^*_l Y_{k}=[(M_{\Delta}L^l)_+,Y_{k}]_-+(k-l+1)Y_{k+l-1},
\end{eqnarray} for $l=0,1,2$ and $k\geq 0$.
\end{lemma}
\begin{proof} A straightforward calculation implies
\begin{eqnarray}\label{partialyk}
\begin{split}
\partial_{l}^* Y_{k}&=\partial_{l}^* \sum_{j=0}^{k-1}[j-\frac{1}{2}(k-1)]L^{k-1-j}(q)\Delta^{-1}(L^*)^j(r)\\
&=\sum_{j=0}^{k-1}[j-\frac{1}{2}(k-1)](\partial_{l}^*(L^{k-1-j}(q))\Delta^{-1}(L^*)^j(r)+L^{k-1-j}(q)\Delta^{-1}(\partial_{l}^*(L^*)^j(r)))\\
&\stackrel{(\ref{lqstar})}{==}\sum_{j=0}^{k-1}[j-\frac{1}{2}(k-1)](M_{\Delta}L^l)_+(L^{k-1-j}(q))\Delta^{-1}(L^*)^j(r)\\
&+\sum_{j=0}^{k-1}[j-\frac{1}{2}(k-1)](k-j-1+\frac{l}{2})L^{k+l-2-j}(q)\Delta^{-1}(L^*)^j(r)\\
&-\sum_{j=0}^{k-1}[j-\frac{1}{2}(k-1)]L^{k-1-j}(q)\Delta^{-1}(M_{\Delta}L^l)^*_+(L^*)^j(r)\\
&+\sum_{j=0}^{k-1}[j-\frac{1}{2}(k-1)](j+\frac{l}{2})L^{k-1-j}(q)\Delta^{-1}(L^*)^{j+l-1}(r).
\end{split}
\end{eqnarray}
The first term and the third one of the right part of
eq.(\ref{partialyk}) can be simplified  to
\begin{eqnarray}\label{13}
\begin{split}
&\sum_{j=0}^{k-1}[j-\frac{1}{2}(k-1)](M_{\Delta}L^l)_+(L^{k-1-j}(q))\Delta^{-1}(L^*)^j(r)\\
&-\sum_{j=0}^{k-1}[j-\frac{1}{2}(k-1)]L^{k-1-j}(q)\Delta^{-1}(M_{\Delta}L^l)^*_+(L^*)^j(r)\\
&=[(M_{\Delta}L^l)_+,Y_{k}]_-.
\end{split}
\end{eqnarray}
Furthermore, on behaving of
\begin{eqnarray*}
&&\sum_{j=0}^{k-1}[j-\frac{1}{2}(k-1)](j+\frac{l}{2})L^{k-1-j}(q)\Delta^{-1}(L^*)^{j+l-1}(r)\\
&&=\sum_{j=l-1}^{k-1+l-1}[j-l+1-\frac{1}{2}(k-1)](j-l+1+\frac{l}{2})L^{k+l-2-j}(q)\Delta^{-1}(L^*)^{j}(r),
\end{eqnarray*}
the second term and the fourth one of the right part of
eq.(\ref{partialyk}) can also be simplified to
\begin{eqnarray} \label{24}
\begin{split}
&\sum_{j=0}^{k-1}[j-\frac{1}{2}(k-1)](k-j-1+\frac{l}{2})L^{k+l-2-j}(q)\Delta^{-1}(L^*)^j(r)\\
&+\sum_{j=0}^{k-1}[j-\frac{1}{2}(k-1)](j+\frac{l}{2})L^{k-1-j}(q)\Delta^{-1}(L^*)^{j+l-1}(r)\\
&=\sum_{j=0}^{k-1}[j-\frac{1}{2}(k-1)](k-j-1+\frac{l}{2})L^{k+l-2-j}(q)\Delta^{-1}(L^*)^j(r)\\
&+\sum_{j=l-1}^{k-1+l-1}[j-l+1-\frac{1}{2}(k-1)](j-l+1+\frac{l}{2})L^{k+l-2-j}(q)\Delta^{-1}(L^*)^{j}(r)\\
&=\sum_{j=0}^{k-1}(k-l+1)[j-\frac{1}{2}(k+l-2)]L^{k+l-2-j}(q)\Delta^{-1}(L^*)^j(r)\\
&=(k-l+1)Y_{k+l-1}.
\end{split}
\end{eqnarray}
Taking (\ref{13}) and (\ref{24}) back into  eq.(\ref{partialyk}) we
have
\begin{eqnarray*}
\partial_{l}^* Y_{k}&=&[(M_{\Delta}L^l)_+,Y_{k}]_-+(k-l+1)Y_{k+l-1}.
\end{eqnarray*}
\end{proof}

Now we are in a position to identity the algebraic structure of the
additional symmetry flows of the cdKP hierarchy.
\begin{theorem}\label{alg}
The additional flows $\partial^*_k$ of the cdKP hierarchy form the
positive half of Virasoro algebra, i.e.,
\begin{equation}
[\partial^*_l,\partial^*_k]=(k-l)\partial^*_{k+l-1},
\end{equation}for $l=0,1,2$, and $k\geq0$.
\end{theorem}
\begin{proof}
\begin{eqnarray*}
[\partial^*_l,\partial^*_k]L
&=&\partial^*_l([-(M_\Delta L^k)_-,L]+[Y_{k-1},L])-\partial^*_k([-(M_\Delta
L^l)_-,L]) \\
&=&\partial_{l}^* [-(M_\Delta L^k)_-,L] +[\partial_{l}^* Y_{k-1},L]+[Y_{k-1}^{(1)},\partial_{l}^* L]+[\partial^*_k(M_\Delta
L^l)_-,L]+[(M_\Delta
L^l)_-,\partial^*_k L]\\
&\stackrel{(\ref{PL})}{==}&[-(\partial_{l}^* (M_\Delta L^k))_-,L] +[-(M_\Delta L^k)_-,(\partial_{l}^* L)] +[[(M_{\Delta}L^l)_+,Y_{k-1}]_-+(k-l)Y_{k+l-2},L]\\
&&+[Y_{k-1},[-(M_{\Delta}L^l)_-,L]]
+[[-(M_{\Delta}L^k)_-+Y_{k-1},M_{\Delta}L^l]_-,L]\\
&&+[(M_\Delta
L^l)_-,[-(M_{\Delta}L^k)_-+Y_{k-1},L]] \\
&=&(k-l)[-(M_{\Delta}L^{k+l-1})_-,L]+(k-l)[Y_{k+l-2},L]\\
&&-[[(M_{\Delta}L^l)_-,Y_{k-1}],L]
+[Y_{k-1},[-(M_{\Delta}L^l)_-,L]]
+[(M_\Delta
L^l)_-,[Y_{k-1},L]]\\
&=&(k-l)\partial^*_{k+l-1}L.
\end{eqnarray*}
$[(M_{\Delta}L^l)_-,Y_{k-1}]_-=[(M_{\Delta}L^l)_-,Y_{k-1}]$
and the Jacobi identity have been used in the fourth identity.
\end{proof}

\section{Conclusions and Discussions}\label{conclusion}

 In this paper, the additional symmetry flows in eq.(\ref{tkflow}) for the cdKP
hierarchy have been constructed by a modification of the
corresponding one of the dKP hierarchy. In this process, the
difference  operator $Y_k$ plays a very crucial role. The
actions of the additional flows (\ref{tkflow}) on the eigenfunction
$q$ and $r$ of one-component cdKP hierarchy were obtained in theorem
\ref{symmetre}. In theorem \ref{alg}, these flows have been shown to
provide a hidden algebraic structure, i.e., the Virasoro
algebra(positive half), in the cdKP hierarchy. Thus we can say  that
the discretization from KP hierarchy to dKP hierarchy is good enough to
retain several interesting mathematical structures including Lax
pair\cite{Kupershimidt}, $\tau$ function\cite{Iliev} and algebraic
structure.

It is possible to extend our results to the $n$-component cdKP
hierarchy associated with a Lax operator
\begin{equation*} \label{laxofncdkp}
L(n)=\Delta+\sum_{i=1}^{m}q_i(n,t)\triangle^{-1}r_i(n,t).
\end{equation*}
Solving the cdKP hierarchy  is also an interesting topic. We shall
do it a near future.

{\bf Acknowledgments} {\noindent \small This work is supported by
the NSF of China under Grant No.10971109 and K.C.Wong Magna Fund in
Ningbo University. Jingsong He is also supported by Program for NCET
under Grant No.NCET-08-0515.  One of the authors (KT) is  supported by Erasmus Mundus Action 2 EXPERTS
 and would like to thank Prof. Jarmo Hietarinta for many helps.
}

%%%%%%%%%%%%%%%%% References  %%%%%%%%%%%%%%%%%%%%%%%%%%%%%%%%%%%%%%%
\newpage{}

%%%%%%%%%%%%%%%%%%%%%%%%%%%%%%%%%%%%%%%%%%%%%%%%%%%%%%%%%

\end{document}